\documentclass[runningheads]{llncs}
\usepackage{graphicx}
\usepackage{tikz}
\usepackage{amsfonts}
\usepackage{float}
\usepackage{xcolor}
\usepackage{amsmath}
\usepackage{hyperref}

\usepackage{algpseudocode}
\usepackage[ruled]{algorithm}
\algrenewcommand\algorithmicrequire{\textbf{Input:}}
\algrenewcommand\algorithmicensure{\textbf{Output:}}
\algtext*{EndWhile}
\algtext*{EndIf}
\algtext*{EndFor}
\algtext*{EndFunction}

\newcommand{\no}[1]{}

\newcommand{\syma}{\texttt{a}}
\newcommand{\symb}{\texttt{b}}
\newcommand{\symc}{\texttt{c}}

\newcommand{\dol}{\texttt{\$}}
\newcommand{\gexp}{\mathtt{exp}}
\newcommand{\bwt}{\mathtt{bwt}}

\newcommand{\floor}[1]{\lfloor #1 \rfloor}

\newcommand{\deltalog}{\delta\log\frac{n\log\sigma}{\delta\log n}}
\newcommand{\dd}{\mathinner{.\,.}}

\newcommand{\suc}{\texttt{arg\_successor}}

\usetikzlibrary{matrix,backgrounds,positioning, arrows}

\begin{document}

\title{Iterated Straight-Line Programs\thanks{Funded with Basal Funds FB0001, ANID, Chile; and ANID-Subdirección de Capital Humano/Doctorado Nacional/2021-21210580.}}

\author{Gonzalo Navarro \and Cristian Urbina}

\institute{CeBiB --- Centre for Biotechnology and Bioengineering \\ Departament of Computer Science, University of Chile}

\maketitle

\begin{abstract}
We explore an extension to straight-line programs (SLPs) that outperforms, for some text families, the measure $\delta$ based on substring complexity, a lower bound for most measures and compressors exploiting repetitiveness (which are crucial in areas like Bioinformatics). The extension, called iterated SLPs (ISLPs), allows rules of the form $A \rightarrow \Pi_{i=k_1}^{k_2} B_1^{i^{c_1}}\cdots B_t^{i^{c_t}}$, for which we show how to extract any substring of length $\lambda$, from the represented text $T[1\dd n]$, in time $O(\lambda + \log^2 n\log\log n)$. This is the first compressed representation for repetitive texts breaking $\delta$ while, at the same time, supporting direct access to arbitrary text symbols in polylogarithmic time. As a byproduct, we extend Ganardi et al.'s technique to balance any SLP (so it has a derivation tree of logarithmic height) to a wide generalization of SLPs, including ISLPs.
\end{abstract}

\section{Introduction}

Motivated by the data deluge, and by the observed phenomenon that many of the fastest-growing text collections are highly repetitive, recent years have witnessed an increasing interest in (1) defining measures of compressibility that are useful for highly repetitive texts, (2) develop compressed text representations whose size can be bounded in terms of those measures, and (3) provide efficient (i.e., polylogarithmic time) access methods to those compressed texts, so that algorithms can be run on them without ever decompressing the texts \cite{Navacmcs20.3,Navacmcs20.2}. We call {\em lower-bounding measures} those satisfying (1), {\em reachable measures} those (asymptotically) reached by the size of a compressed representation (2), and {\em accessible measures} those reached by the size of representations satisfying (3).

For example, the size $\gamma$ of the smallest ``string attractor'' of a text $T$ is a lower-bounding measure, unknown to be reachable \cite{KP18}, and smaller than the size reached by known compressors. The size $b$ of the smallest ``bidirectional macro scheme'' of $T$ \cite{SS82}, and the size $z$ of the Lempel-Ziv parse of $T$ \cite{LZ76}, are reachable measures. The size $g$ of the smallest context-free grammar generating (only) $T$ \cite{CLLPPSS05} is an accessible measure \cite{BLRSRW15}. It holds $\gamma \le b \le z \le g$ for every text.

One of the most attractive lower-bounding measures devised so far is $\delta$ \cite{RRRS13,CEKNP20}. Let $T[1\dd n]$ be a text over alphabet $[1\dd\sigma]$, and $T_k$ be the number of distinct substrings of length $k$ in $T$, which define its so-called substring complexity. Then the measure is $\delta(T) = \max_k T_k/k$. This measure has several attractive properties: it can be computed in linear time and lower-bounds all previous measures of compressibility, including $\gamma$, for every text. While $\delta$ is known to be unreachable, the measure $\delta'=\deltalog$ has all the desired properties: $\Omega(\delta')$ is the space needed to represent some text family for each $n$, $\sigma$, and $\delta$; within $O(\delta')$ space it is possible to represent every text $T$ and access any length-$\lambda$ substring of $T$ in time $O(\lambda+\log n)$ \cite{KNP22}, together with more powerful operations \cite{KNP22,KNO23,KK23}. 

As for $g$, a {\em straight-line program (SLP)} is a context-free grammar that generates (only) $T$, and has size-2 rules of the form $A \rightarrow BC$, where $B$ and $C$ are nonterminals, and size-1 rules $A \rightarrow \syma$, where $\syma$ is a terminal symbol. The SLP size is the sum of all its rule sizes. A {\em run-length SLP (RLSLP)} may contain, in addition, size-2 rules of the form $A \rightarrow B^k$, representing $k$ repetitions of nonterminal $B$. A RLSLP of size $g_{rl}$ can be represented in $O(g_{rl})$ space, and within that space we can offer fast string access and other operations
\cite[App.~A]{CEKNP20}. It holds $\delta \le g_{rl} = O(\delta')$, where $g_{rl}$ is the smallest RLSLP that generates $T$ \cite{Navacmcs20.3,KNP22} (the size $g$ of the smallest grammar or SLP, instead, is not always $O(\delta')$). 

While $\delta$ lower-bounds all previous measures on every text, $\delta'$ is not the smallest accessible measure. 
In particular, $g_{rl}$ is always $O(\delta')$, and it can be smaller by up to a logarithmic factor. Indeed, $g_{rl}$ is a minimal accessible measure as far as we know. It is asymptotically between $z$ and $g$ \cite{Navacmcs20.3}. An incomparable accessible measure is $z_{end} \ge z$, the size of the LZ-End parse of the text \cite{KN13,KS22}.

The belief that $\delta$ is a lower bound to every reachable measure was disproved by the recently proposed L-systems \cite{NU21,NU23}. L-systems are like SLPs where all the symbols are nonterminals and the derivation ends at a specified depth in the derivation tree. The size $\ell$ of the smallest L-system generating $T[1\dd n]$ is a reachable measure of repetitiveness and was shown to be as small as $O(\delta/\sqrt{n})$ on some text families, thereby sharply breaking $\delta$ as a lower bound. Measure $\ell$, however, is unknown to be accessible, and thus one may wonder whether there exist accessible text representations that are smaller than $\delta$.

In this paper we devise such a representation, which we call {\em iterated SLP (ISLP)}. ISLPs extend SLPs (and RLSLPs) by allowing a more complex version of the rule $A \rightarrow B^k$, namely $A \rightarrow \Pi_{i=k_1}^{k_2} B_1^{i^{c_1}}\cdots B_t^{i^{c_t}}$ of size $2+2t$. We show how to extract a substring of length $\lambda$ from the ISLP of a text $T$ in time $O(\lambda+\log^2 n \log\log n)$ provided the ISLP is balanced, that is, its derivation tree is of height $O(\log n)$.
Just like SLPs and RLSLPs can be balanced \cite{GJL2021,NOU2022} while retaining their asymptotic size, we show how to balance a more general class of SLP extensions we call generalized SLPs (GSLPs). GSLPs, which include ISLPs, allow rules of the form $A \rightarrow x$, where $x$ is a {\em program} that outputs the right-hand side of the rule. We show that, if every nonterminal appearing in $x$'s output does it at least twice, then the GSLP can be balanced in the same way as SLPs. This byproduct of our results can be of independent interest to provide polylogarithmic-time access to other extensions of context-free grammars.

\section{Preliminaries}

We explain some concepts and notation used in the rest of the paper.

\subsubsection{Strings}

Let $\Sigma = [1\dd \sigma]$ be an \emph{alphabet}. A \emph{string} $T[1\dd n]$ of length $n$ is a finite sequence $T[1]\,T[2] \dots T[n]$ of $n$ symbols in $\Sigma$. We denote by $\varepsilon$ the unique string of length $0$. We denote by $\Sigma^*$ the set of all finite strings with symbols in $\Sigma$. The $i$-th symbol of $T$ is denoted by $T[i]$, and the sequence $T[i]\dots T[j]$ is denoted by $T[i\dd j]$. The \emph{concatenation} of $X[1\dd n]$ and $Y[1\dd m]$ is defined as $X \cdot Y = X[1] \dots X[n]\, Y[1] \dots Y[m]$ (we omit the dot when there is no ambiguity). If $T = XYZ$, then $X$ (resp. $Y$, resp. $Z$) is 
 a \emph{prefix} (resp. \emph{substring}, resp. \emph{suffix}) of $T$. A \emph{power} $T^k$ stands for $k$ consecutive concatenations of the string $T$.
 We denote by $|T|_a$ the number of occurrences of the symbol $a$ in $T$. A \emph{string morphism} is a function $\varphi: \Sigma^* \rightarrow \Sigma^*$ such that $\varphi(xy) = \varphi(x) \cdot \varphi(y)$ for any strings $x$ and $y$.

\subsubsection{Straight-Line Programs}

A \emph{straight-line program} (SLP) is a context-free grammar \cite{Sipser2012} that contains only terminal rules of the form $A \rightarrow \syma$ with $\syma \in \Sigma$, and binary rules of the form $A \rightarrow BC$ for variables $B$ and $C$ whose derivations cannot reach again $A$. These restrictions ensure that each variable of the SLP generates a unique string, defined as $\gexp(A) = \syma$ for a rule $A \rightarrow \syma$, and as $\gexp(A) = \gexp(B)\cdot\gexp(C)$ for a rule $A \rightarrow BC$. A \emph{run-length straight-line program} (RLSLP) is an SLP that also admits run-length rules of the form $A \rightarrow B^k$ for some $k \geq 3$, with their expansion defined as $\gexp(A)=\gexp(B)^k$. The {\em size} of an SLP is the sum of the lengths of the right-hand sides of its rules; the size of an RLSLP is defined similary, assuming that rules $A \rightarrow B^k$ are of size 2 (i.e., two integers to represent $B$ and $k$).

The \emph{derivation tree} of an SLP is an ordinal tree where the nodes are the variables, the root is the initial variable, and the leaves are the terminal variables. The children of a node are the variables appearing in the right-hand side of its rule (in left-to-right order). The \emph{height} of an SLP is the length of the longest path from the root to a leaf node in the derivation tree. The height of an RLSLP is obtained by \emph{unfolding} its run-length rules, that is, writing a rule $B^k$ as $BB\dots B$ where $B$ appears $k$ times, to obtain an equivalent SLP (actually, a slight extension where the right-hand sides can feature more than two variables).

SLPs and RLSLPs yield measures of repetitiveness $g$ and $g_{rl}$, defined as the size of the smallest SLP and RLSLP generating the text, respectively. Clearly, it holds that $g_{rl} \le g$. It also has been proven that $g$ is NP-hard to compute \cite{CLLPPSS05}. 

\subsubsection{Other Repetitiveness Measures}

For self-containedness, we describe the most important repetitiveness measures and relate them with the accessible measures $g$ and $g_{rl}$; for more details see a survey \cite{Navacmcs20.3}.

\paragraph{Burrows-Wheeler Transform.}

The \emph{Burrows-Wheeler Transform} (BWT) \cite{BW94} is a reversible permutation of $T$, which we denote by $\bwt(T)$. It is obtained by sorting lexicographically all the rotations of the string $T$ and concatenating their last symbols, which can be done in $O(n)$ time. The measure $r$ is defined as the size of the \emph{run-length encoding} of $\bwt(T)$. Usually, $T$ is assumed to be appended with a sentinel symbol $\dol$ strictly smaller than any other symbol in $T$, and then we call $r_\dol$ the size of the run-length encoding of $\bwt(T\dol)$. This measure is then reachable, and fully-functional indexes of size $O(r_\dol)$ exist \cite{GNP18}, but interestingly, it is unknown to be accessible. While this measure is generally larger than others, it can be upper-bounded by $r_\dol = O(\delta\log\delta\log\frac{n}{\delta})$ \cite{KK20}.

\paragraph{Lempel-Ziv Parsing.}

The \emph{Lempel-Ziv parsing} (LZ) \cite{LZ76} of a text $T[1\dd n]$ is a \emph{factorization} into non-empty \emph{phrases}  $T = X_1X_2\dots X_z$ where each $X_i$ is either the first occurrence of a symbol or the longest prefix of $X_i \dots X_z$ with a copy in $T$ starting at a position in  $[1\dd |X_1\dots X_{i-1}|]$. LZ is called a \emph{left-to-right} parsing because each phrase has its \emph{source} starting to the left, and it is optimal among all parsings satisfying this condition. It can be constructed greedily from left to right in $O(n)$ time. The measure $z$ is defined as the number of phrases in the LZ parsing of the text, and it has been proved that $z \le g_{rl}$ \cite{NOPtit20}. While $z$ is obviously reachable, it is unknown to be accessible. A close variant $z_{end} \ge z$ \cite{KN13} that forces phrase sources to be end-aligned with a preceding phrase, has been shown to be accessible \cite{KS22}. 

\paragraph{Bidirectional macro schemes.}

A \emph{bidirectional macro scheme} (BMS) \cite{SS82} is a factorization of a text $T[1\dd n]$ where each phrase can have its source starting either to the left or to the right. The only requeriment is that by following the pointers from phrases to sources, we should eventually be able to fully decode the text. The measure $b$ is defined as the size of the smallest BMS representing the text. Clearly, $b$ is reachable, but it is unknown to be accessible. It holds that $b \le z$, and it was proved that $b \le r_\$$ \cite{NOPtit20}. Computing $b$ is NP-hard \cite{Gallant1982}. 

\paragraph{String Attractors.}

A \emph{string attractor} for a text $T[1\dd n]$ is a set of positions $\Gamma \subseteq [1\dd n]$ such that any substring of $T[i\dd j]$ has an occurrence $T[i'\dd j']$ crossing at least one of the positions in $\Gamma$ (i.e., there exist $k \in \Gamma$ such that $i'\le k\le j'$). The measure $\gamma$ is defined as the size of the smallest string attractor for the string $T$, and it is NP-hard to compute \cite{KP18}. It holds that $\gamma$ lower bounds the size $b$ of the smallest bidirectional macro scheme and can sometimes be asymptotically smaller \cite{BFIKMN21}. On the other hand, it is unknown if $\gamma$ is reachable.

\paragraph{Substring Complexity.}

Let $T[1\dd n]$ be a text and $T_k$ be the number of distinct substrings of length $k$ in $T$, which define its so-called substring complexity. Then the measure is $\delta = \max_k T_k/k$ \cite{RRRS13,CEKNP20}. This measure can be computed in $O(n)$  time and lower-bounds $\gamma$, and thus all previous measures of compressibility, for every text. On the other hand, it is known to be unreachable \cite{KNP22}. The related measure $\delta'=\deltalog$ is reachable and accessible, and still lower-bounds $b$ and all other reachable measures on some text family for every $n$, $\sigma$, and $\delta$ \cite{KNP22}. Besides, $g_{rl}$ (and thus $z$, $b$, and $\gamma$, but not $g$) are upper-bounded by $O(\deltalog)$; $g$ can be upper-bounded by $O(\gamma\log^2\frac{n}{\gamma})$ \cite{KNP22,KP18}.

\paragraph{L-systems.}

An \emph{L-system} (for compression) is a tuple $L = (V, \varphi, \tau, S, d, n)$ extending a traditional Lindenmayer system \cite{Lindenmayer1968a,Lindenmayer1968b}, where $V$ is the set of variables (which are also considered as terminal symbols), $\varphi: V \rightarrow V^+$ is the set of rules (and also a morphism of strings), $\tau: V \rightarrow V$ is a coding, $S \in V$ the initial variable, and $d$ and $n$ are integers. The string generated by the system is $\tau(\varphi^d(S))[1\dd n]$. The measure $\ell$ is defined as the size of the smallest L-system generating the string. It has been proven that $\ell$ is incomparable to $\delta$ ($\ell$ can be smaller by a $\sqrt{n}$ factor) and almost any other repetitiveness measure considered in the literature \cite{NU21,NU23}.

\section{Iterated Straight-Line Programs}

We now define iterated SLPs and show that they can be much smaller than $\delta$.

\begin{definition}
An \emph{iterated straight-line program} of \emph{degree} $d$ ($d$-ISLP)  is an SLP that allows in addition \emph{iteration rules} of the form $$A \rightarrow \prod_{i=k_1}^{k_2} B_1^{i^{c_1}} \cdots B_t^{i^{c_t}}$$ 
where $1 \le k_1, k_2$, $0 \le c_1,\dots,c_t \leq d$ are integers and $B_1 \dots B_t$ are variables that cannot reach $A$ (so the ISLP generates a unique string). Iteration rules have size $2+2t=O(t)$ and expand to
$$\gexp(A) = \prod_{i=k_1}^{k_2} \gexp(B_1)^{i^{c_1}}\!\cdots \gexp(B_t)^{i^{c_t}}$$
where if $k_1 > k_2$ the iteration goes from $i=k_1$ downwards to $i=k_2$.
The size $size(G)$ of a $d$-ISLP $G$ is the sum of the sizes of all of its rules.
\end{definition}

\begin{definition}
The measure $g_{it(d)}(T)$ is defined as the size of the smallest $d$-ISLP that generates $T$, whereas $g_{it}(T) = \min_{d \ge 0} g_{it(d)}(T)$.
\end{definition}

The following observations show that ISLPs subsume RLSLPs, and thus, can be smaller than the smallest L-system.

\begin{proposition}For any $d \geq 0$, it always holds that $g_{it(d)} \le g_{rl}$.
\end{proposition}

\begin{proof}Just note that a rule $A \rightarrow \prod_{i=1}^k B^{i^0}$ from an ISLP simulates a rule $A \rightarrow B^k$ from a RLSLP. In particular, $0$-ISLPs are equivalent to RLSLPs. \qed
\end{proof}

\begin{proposition}For any $d \geq 0$, there exists a string family where $g_{it(d)} = o(\ell)$.
\end{proposition}

\begin{proof}Navarro and Urbina show a string family where $g_{rl} = o(\ell)$ \cite{NU23}. Hence, $g_{it(d)}$ is also $o(\ell)$ in this family. \qed
\end{proof}

We now show that $d=1$ suffices to obtain ISLPs that are significantly smaller than $\delta$ for some string families.

\begin{lemma}Let $d \ge 1$. There exists a string family with $g_{it(d)} = O(1)$ and $\delta = \Omega(\sqrt{n})$.
\end{lemma}

\begin{proof}Such a family is formed by the strings $s_k = \prod_{i=1}^k \syma^i\symb$. The 1-ISLPs with initial rule $S_k \rightarrow \prod_{i=1}^{k} A^iB$, and rules $A \rightarrow \syma$, $B \rightarrow \symb$, generate each string $s_k$ in the family using $O(1)$ space. On the other hand, it has been proven that $\delta = \Omega(\sqrt{n})$ in the family $\symc s_k$ \cite{NU23}. As $\delta$ can only decrease by $1$ after the deletion of a character \cite{AFI2023},  $\delta = \Omega(\sqrt{n})$ in the family $s_k$ too.\qed
\end{proof}

On the other hand, ISLPs can perform worse than other compressed representations; recall that $\delta \le \gamma \le b \le r_\dol$.

\begin{lemma}\label{lem:fib} Let  $\mu \in \{r, r_\dol, \ell\}$. There exists a string family with $g_{it(d)} = \Omega(\log n)$ and $\mu = O(1)$.
\end{lemma}

\begin{proof}
Consider the family of Fibonacci words defined recursively as $F_0 = \syma$, $F_1 = \symb$, and $F_{i+2} = F_{i+1}F_i$ for $i \geq 0$. Fibonacci words cannot contain substrings of the form $x^4$ for any $x \neq \varepsilon$ \cite{KARHUMAKI1983}. Consider an ISLP for a Fibonacci word and a rule of the form $A \rightarrow \prod_{i=k_1}^{k_2}B_1^{i^{c_1}}\cdots B_t^{i^{c_t}}$. Observe that if $c_r \neq 0$ for some $r$, then $\max(k_1, k_2) < 4$, as otherwise $\gexp(B_r)^4$ occurs in $T$. Similarly, if $c_r = 0$ for all $r$, then $|k_1 - k_2| < 3$, as otherwise $\gexp(B_1\cdots B_t)^4$ appears in $T$. In the latter case, we can rewrite the product with $k_1, k_2 \in [1\dd 3]$. Therefore, we can unfold the product rule into standard SLP rules of total size at most $9t$ ($3t$ variables raised to at most 3 each). Hence, for any $d$-ISLP $G$ generating a Fibonacci word, there is an SLP $G'$ of size $O(|G|)$ generating the same string. As $g = \Omega(\log n)$ in every string family \cite{Navacmcs20.3}, we obtain that $g_{it(d)} = \Omega(\log n)$ in this family too. On the other hand, $r_\dol, r$, and $\ell$ are $O(1)$ in the even Fibonacci words \cite{NOPtit20,MRS2003,NU21}.\qed
\end{proof}

\begin{lemma}There exists a string family satisfying that $z = O(\log n)$ and $g_{it(d)} = \Omega(\log^2 n/\log\log n)$.
\end{lemma}

\begin{proof}Let $T(n)$ be the length $n$ prefix of the infinite Thue-Morse word on the alphabet $\{\syma,\symb\}$. Let $k_1,...,k_p$ be a set of distinct positive integers, and consider strings of the form $S = T(k_1)|_1T(k_2)|_2\cdots T(k_{p-1})|_{p-1}T(k_p)$, where $|_i$'s are unique separators and $k_1$ is the largest of the $k_i$. Since the sequences $T(k_i)$ are cube-free \cite{AlloucheShallit_ThueMorse}, there is no asymptotic difference in the size of the smallest SLP and the smallest ISLP (similarly to Lemma \ref{lem:fib}) for the string $S$. Hence, $g_{id(d)} = \Theta(g)$ in this family. It has been proven that $g = \Omega(\log^2 k_1 /\log\log k_1)$ and $z = O(\log k_1)$ for some specific sets of integers where $p = \Theta(k_1)$ \cite{BGLP2018}. Thus, the result follows.\qed
\end{proof}

One thing that makes ISLPs robust is that they are not very sensitive to reversals, morphism application, or edit operations (insertions, deletions, and substitutions of a single character). The measure $g_{it(d)}$ behaves similarly to SLPs in this matter, for which it has been proved that $g(T') \le 2g(T)$ after an edit operation that converts $T$ to $T'$ \cite{AFI2023}, and that $g(\varphi(T)) \le g(T) + c_{\varphi}$ with $c_{\varphi}$ a constant depending only on the morphism $\varphi$ \cite{FRMU2023}. This makes $g_{it(d)}$ much more robust to string operations than measures like $r$ and $r_\$$, which are highly sensitive to all these transformations \cite{GILPST2021,GILRSU2023,FRMU2023,AFI2023}.

\begin{lemma} \label{lem:edit}
Let $G$ be a $d$-ISLP generating $T$. Then there exists a $d$-ISLP of size $|G|$ generating the reversed text $T^R$. Let $\varphi$ be a morphism. Then there exists a $d$-ISLP of size $|G| + c_\varphi$ generating the text $\varphi(T)$, where $c_\varphi$ is a constant depending only on $\varphi$. Moreover, there exists a $d$-ISLP of size at most $O(|G|)$ generating $T'$ where $T$ and $T'$ differ by one edit operation.
\end{lemma}

\begin{proof}
We omit the proof for the first two claims as they are fairly easy to see. For the edit operations, we proceed as follows. Consider the derivation tree of the ISLP, and the path from the root to the character we want to substitute, delete, or insert a character before or after. Then, we follow this path in a bottom up manner, constructing a new variable $A'$ for each node $A$ we visit. We start at some $A \rightarrow \syma$, so we construct $A' \rightarrow x$ where either $x = \symc$ or $x = \syma\symc$ or $x= \symc\syma$ or $x = \varepsilon$ depending on the edit operation. If we reach a node $A \rightarrow BC$ going up from $B$ (so we already constructed $B'$), we construct a node $A' \rightarrow B'C$ (analogously if we come from $C$).
If we reach a node $A \rightarrow \prod_{i=k_1}^{k_2} B_1^{i^{c_1}}\dots B_t^{i^{c_t}}$ going up from a specific $B_r$ with $r \in [1\dd t]$ (so we already constructed $B_r'$) at the $k$-th iteration of the product with $k_1 \le k \le k_2$ and being the $q$-th copy of $B_r$ inside $B_r^{k^{c_r}}$, then we construct the following new rules 
\begin{align*}
&A_1  \rightarrow \prod_{i=k_1}^{k-1}B_1^{i^{c_1}}\dots B_t^{i^{c_t}},\, A_2  \rightarrow \prod_{i=k}^{k}B_1^{i^{c_1}}\dots B_{r-1}^{i^{c_{r-1}}},\, A_3  \rightarrow \prod_{i=1}^{{q-1}}B_r^{i^0}, \\ 
&A_4  \rightarrow \prod_{i=q+1}^{k^{c_r}}B_r^{i^0},\,
A_5  \rightarrow \prod_{i=k}^{k}B_{r+1}^{i^{c_{r+1}}}\dots B_t^{i^{c_t}},\,
A_6  \rightarrow \prod_{i=k+1}^{k_2}B_1^{i^{c_1}}\dots B_t^{i^{c_t}}\\
&A' \rightarrow A_1A_2A_3B_r'A_4A_5A_6
\end{align*}
which are equivalent to $A$ (except by the modified, inserted, or deleted symbol) and sum to a total size of at most $6t + 21$. As $t \geq 1$, it holds that $(6t + 21)/(2t+2) \le 7$. After finishing the whole process, we obtain a $d$-ISLP of size at most $8|G|$. Note that this ISLP contains $\varepsilon$-rules. It also contains some non-binary SLP rules, which can be transformed into binary rules, at most doubling the size of the grammar.  
\qed
\end{proof}




\section{Accessing ISLPs}

We have shown that $g_{it(d)}$ breaks the lower bound $\delta$ already for $d \ge 1$. We now show that the measure is accessible. Concretely, we will show that any substring of length $\lambda$ can be extracted in time $O(\lambda + (h+\log n)\,\log n\log\log n)$, where $h$ is the height of the grammar tree, and in Section~\ref{sec:bal} we show that ISLPs can be balanced so they have $h=O(\log n)$. In total, we obtain the following result.

\begin{theorem}
Let $T[1\dd n]$ be represented by a $d$-ISLP of size $g_{it}$. Then, there exists a data structure of size $O(g_{it})$ that extracts any substring of $T$ of length $\lambda$ in time $O(\lambda + \log^2 n\log\log n)$ on a RAM machine of $\Theta(\log n)$ bits, using $O(\log^2 n \log \log n)$ additional words of working space.
\end{theorem}

In fact, our extraction time is $O(\lambda + d\log d\log n + d^2\log d)$ using $O(d^2 \log d)$ working space, which reduces to $O(\lambda+\log n)$ time and $O(1)$ working space for $d=O(1)$ (recall that 1-ISLPs already break the $\delta$ lower-bound), and yields the result in the theorem if $d=O(\log n)$. For larger $d$, we start with a technical result that shows that we can always force $d$ to be $O(\log n)$ without asymptotically increasing the size. From now on in the paper, we will disregard for simplicity the case $k_1 > k_2$ in the rules $A \rightarrow \Pi_{i=k_1}^{k_2} B_1^{i^{c_1}}\cdots B_t^{i^{c_t}}$, as their treatment is analogous to that of the case $k_1 \le k_2$.

\begin{lemma} \label{lem:dlog}
If a $d$-ISLP $G$ generates $T[1\dd n]$, then there is also a $d'$-ISLP $G'$ of the same size that generates $T$, for some $d' \le \log_2 n$.
\end{lemma}
\begin{proof}
For any rule $A = \prod_{i=k_1}^{k_2} B_1^{i^{c_1}}\cdots B_t^{i^{c_t}}$, any $i \in [k_1\dd k_2]$, and any $c_j$, it holds that $n \ge |\gexp(A)| \ge i^{c_j}$, and therefore $c_j \le \log_i n$, which is bounded by $\log_2 n$ for $i \ge 2$. Therefore, if $k_2 \ge 2$, all the values $c_j$ can be bounded by some $d' \le \log_2 n$. A rule with $k_1=k_2=1$ is the same as $A \rightarrow B_1\cdots B_t$, so all values $c_j$ can be set to $0$ without changing the size of the rule at all.\qed
\end{proof}

\subsection{Data Structures}

We define some data structures that extend ISLPs allowing us to efficiently navigate it within $O(g_{it}$) space. Per Lemma~\ref{lem:dlog}, we assume $d = O(\log n)$.

Consider a rule $A \rightarrow \prod_{i=k_1}^{k_2} B_1^{i^{c_1}}\dots B_t^{i^{c_t}}$. Though $t$ can be large, there are only $d+1$ distinct values $c_j$. We will make use of auxiliary polynomials 
$$f_r(i) ~=~ \sum_{j=1}^{r}|\gexp(B_j)|\cdot i^{c_j},$$ 
for $r \in [1,t]$, to navigate within the ``blocks'' $i$: $f_r(i)$ computes cumulative lengths inside the product expression $B_1^{i^{c_1}}\dots B_t^{i^{c_t}}$, up to the variable $B_r$, for a given $i$. 

We now show how to compute any $f_r(i)$ in time $O(d)$ using $O(t)$ space for each $A$.
An array $S_A[1\dd t]$ stores cumulative length information, as follows
$$ S_A[r] = \sum_{1 \le j \le r, c_j = c_r}  |\gexp(B_j)|.$$
That is, $S_A[r]$ adds up the lengths of the symbol expansions up to $B_r$ that must be multiplied by $i^{c_r}$. A second array, $C_A[1\dd t]$, stores the values $c_1,\ldots,c_t$. We preprocess $C_A$ to solve predecessor queries of the form
$$ pred(A,r,c) = \max \{ j \le r,~C_A[j] = c\},$$
that is, the latest occurrence of $c$ in $C_A$ to the left of position $r$, for every $c=0,\ldots,d$. This query can be answered in $O(d)$ time because the elements in $C_A$ are also in $\{ 0,\ldots,d \}$: cut $C_A$ into chunks of length $d+1$, and for each chunk $C_A[(d+1)\cdot j +1 \dd (d+1)\cdot(j+1)]$ store precomputed values $pred(A,(d+1)\cdot j,c)$ for all $c \in \{0,\ldots, d\}$. This requires $O(t)$ space. To compute the values $r_c = pred(A,r,c)$ for all $c$, find the chunk $j = \lceil r/(d+1) \rceil-1$ where $r$ belongs, initialize every $r_c = pred(A,(d+1)\cdot j,c)$ for every $c$ (which is stored with the chunk $j$), and then scan the chunk prefix $C_A[(d+1)\cdot j+1 \dd r]$ left to right, correcting every $r_c \gets k$ if $c=C_A[k]$, for $k=(d+1)\cdot j+1 \dd r$.

 We can then evaluate $f_r(i)$ in $O(d)$ time by computing all values $r_c$ as explained (i.e., the last position to the left of $r$ where the exponent is $c$), and adding up $S_A[r_c]\cdot i^c$ (because $S_A[r_c]$ adds up all $|\gexp(B_j)|$ that must be multiplied by $i^c$ in $f_r(i)$).
We also define the polynomial 
$$f^+(k) ~=~ \sum_{i = k_1}^{k}f_t(i)$$ 
to select a ``block'': $f^+(k)$ computes the cumulative sum of the length of the whole expressions $B_1^{i^{c_1}}\cdots B_t^{i^{c_t}}$ until $i=k$. Note we cannot afford to store all the $k_2-k_1+1$ values $f^+(k)$, but we can exploit the fact that the polynomials $f_t(i)$ have degree at most $d$, and thus $f^+(k)$ is a polynomial on $k$ of degree at most $d+1$. Storing $f^+$ as a polynomial, then, requires only $O(d)$ space, instead of the $O(k)$ space needed to store all of its values. This can still be excessive, however, as it blows the space by an $O(\log n)$ factor in a rule like $A \rightarrow \Pi_{i=k_1}^{k_2} B^{i^d}$, which is of size $4$ but $f^+$ is of degree $d+1$. 

We will instead compute $f^+(k)$ in $O(d)$ arithmetic operations by reusing the same data structures we store for $f_r(i)$: for each $c=0,\ldots,d$, we compute $t_c = pred(A,t,c)$ and $s_c = S_A[t_c]$. Instead of accumulating $s_c \cdot i^c$, however, we accumulate $s_c \cdot \sum_{i=k_1}^{k} i^c = s_c \cdot (p_c(k)-p_c(k_1-1))$, where $p_c(k) = \sum_{i=1}^{k} i^c$. 

We cannot afford storing all the $O(kd)$ values $p_c(k)$, but since there are only $d+1=O(\log n)$ functions $p_c$ and each one is a polynomial of degree $c+1=O(\log n)$, they can be represented as polynomials using $O(\log^2 n)$ integers. Further, they can be computed at query time\footnote{Indeed, the polynomials $p_c(k)$ are independent of the grammar, so they can be computed once for all queries and for all grammars.}, before anything else, in $O(d^2)$ arithmetic operations using, for each $c$, the formula\footnote{See Wolfram Mathworld's {\tt https://mathworld.wolfram.com/BernoulliNumber.html}, Eqs.~(34) and (47).}
$$ p_c(k) ~~=~~ k^c + \frac{1}{c+1} \cdot \sum_{j=0}^c {c+1 \choose j}\, B_j \cdot k^{c+1-j},$$
which is a polynomial on $k$ of degree at most $d+1$. The formula requires $O(c)$ arithmetic operations once the numbers $B_j$ are computed. Those $B_j$ are the Bernoulli (rational) numbers.  All the Bernoulli numbers from $B_0$ to $B_d$ can be computed in $O(d^2)$ arithmetic operations using the recurrence
$$ \sum_{j=0}^{d} {d+1 \choose j}\, B_j ~=~ 0,$$
from $B_0=1$. The numerators and denominators of the rationals $B_j$ fit in $O(j \log j) = O(d \log d) = O(\log n\log d)$ bits,\footnote{See {\tt https://www.bernoulli.org}, sections ``Structure of the denominator'', ``Structure of the nominator'', and ``Asymptotic formulas''.}
so they can be operated in $O(\log d)$ time in a RAM machine with word size $\Theta(\log n)$. Therefore, the total preprocessing time to later compute any $f^+(k)$ is $O(d^2 \log d)$. We note, however, that due to the length of the numerators and denominators of the fractional Bernoulli numbers, the time to compute any $f^+(k)$ is $O(d\log d)$.


\begin{example}
Consider the ISLP of Proposition 2, defined by the rules
$S \rightarrow \prod_{i=1}^{k_2} A^iB$, $A \rightarrow \syma$, and $B \rightarrow \symb$. The polynomials associated with the representation of the rule $S$ are $i^{c_1} = i$ and $i^{c_2} = 1$. Then, we construct the auxiliary polynomials $f_1(i) = |\gexp(A)|i^{c_1} = i$ and $f_2(i) = |\gexp(A)| i^{c_1}+ |\gexp(B)| i^{c_2} = i + 1$. Finally, we construct the auxiliary polynomial $f^+(k) = \sum_{i=1}^k f_2(i) = \sum_{i=1}^k (i+1) = \frac{1}{2} k^2 + \frac{3}{2}k$. Figure~\ref{fig:polys} shows a more complex example to illustrate $C_A$ and $S_A$.
\end{example}

\begin{figure}[t]
\centering
\begin{tikzpicture}[array/.style={matrix of nodes,nodes={draw, minimum size=7mm, fill=gray!30},column sep=-\pgflinewidth, row sep=1mm, nodes in empty cells,
row 1/.style={nodes={draw=none, fill=none, minimum size=7mm}}}]
\matrix[array] (array) {
1 & 2 & 3 & 4 & 5 & 6 & 7 & 8\\
2  & 3  & 6  & 7 & 14 & 13 & 5 & 3\\
1  & 2  & 1  &  0 & 0 & 1 & 2 & 3\\};
\draw (array-2-1.west)--++(0:0mm) node [left] (sa) {$S_A$};
\draw (array-3-1.west)--++(0:0mm) node [left] {$C_A$};
\draw (array-2-8.east)--++(0:0mm) node [right] {$\hspace{0.5cm}f_8(i) = 3i^3 + 5i^2 + 13i + 14$};
\draw (array-3-8.east)--++(0:0mm) node [right] (empty2) {$\hspace{0.5cm}f^+(k) = \frac{9}{12}k^4 + \frac{38}{12}k^3 + \frac{117}{12}k^2 + \frac{256}{12}k$};
\end{tikzpicture}
\caption{Data structures built for the ISLP rule $A \rightarrow \prod_{i=1}^5B^{i}C^{i^2}D^{i}EEE^iB^{i^2}C^{i^3}$, with $|\gexp(B)| = 2$,  $|\gexp(C)| = 3$,  $|\gexp(D)| = 4$, and $|\gexp(E)| = 7$. We show some of the polynomials to be simulated with these data structures. }
\label{fig:polys}
\end{figure}

\subsection{Direct Access in Time $O((h+\log n)\,d\log d)$}

We start with the simplest query: given the data structures of size $O(g_{it})$ defined in the previous sections, return the symbol $T[l]$ given an index $l$.

For SLPs with derivation tree of height $h$, the problem is easily solved in $O(h)$ time by storing the expansion size of every nonterminal, and descending from the root to the corresponding leaf using $|\gexp(B)|$ to determine whether to descend to the left or to the right of every rule $A \rightarrow BC$. This is easy to generalize in RLSLP rules $A \rightarrow B^k$, because every repetition corresponds to the same string, of length $|\gexp(B)|$. The general idea for $d$-ISLPs is similar, but now determining which child to follow in repetition rules is more complex. 

To access the $l$-th character of the expansion of $A \rightarrow \prod_{i=k_1}^{k_2} B_1^{i^{c_1}}\cdots B_t^{i^{c_t}}$ we first find the value $i$ such that $f^+(i-1) < l \le f^+(i)$ by using binary search. Then, we find the value $r$ such that $f_{r-1}(i) < l -  f^+(i-1) \le f_r(i)$ by using binary search in the subindex of the polynomials. We then know that the search follows by $B_r$, with offset $l -  f^+(i-1) - f_{r-1}(i)$ inside $|\gexp(B_r)|^{i^{c_r}}$. The offset within $B_r$ is then easily computed with a modulus, as in RLSLPs. 
Algorithm~\ref{alg:direct-access} gives the details.

We carry out the first binary search so that, for every $i$ we try, if $f^+(i) < l$ we immediately answer $i+1$ if $l \le f^+(i+1)$; instead, if $l \le f^+(i)$, we immediately answer $i$ if $f^+(i-1) < l$. As a result, the search area is initially of length $|\gexp(A)|$ and, if the answer is $i$, the search has finished by the time the search area is of length $\le f^+(i)-f^+(i-1) = f_t(i)$. Thus, there are $O(1+\log(|\gexp(A)|/f_t(i)))$ binary search steps. The second binary search is modified analogously so that it carries out $O(1+\log(f_t(i)/(i^{c_r} |\gexp(B_r)|)))$ steps, for a total of at most $O(1+\log(|\gexp(A)|/|\gexp(B_r)|))$ steps. As the search continues by $B_r$, the sum of binary search steps telescopes to $O(h+\log n)$ on an ISLP of height $h$, and the total time is $O((h+\log n)\,d\log d)=O((h+\log n)\log n\log\log n)$. 


\begin{algorithm}[t]\caption{Direct access for ISLPs in $O((h+\log n)\, d\log d)$ time}\label{alg:direct-access}
\begin{algorithmic}[1]
\Require An ISLP $G$ of height $h$, a variable $A$ of $G$, and a position $l \in [1, |\gexp(A)|]$.
\Ensure The character $\gexp(A)[l]$ at position $l$ in $\gexp(A)$.
\Function{access}{$G, A, l$}
\If{$A \rightarrow a$}
    \State \Return $a$
\EndIf
\If{$A \rightarrow BC$}
    \If{$l \le |\gexp(B)|$}
        \State \Return \Call{access}{$G, B, l$}
    \Else
        \State \Return \Call{access}{$G, C, l-|\gexp(B)|$}
    \EndIf
\EndIf
\If{$A \rightarrow \prod_{i=k_1}^{k_2} B_1^{i^{c_1}} \dots B_t^{i^{c_t}}$}
\State $i \leftarrow \suc([f^+(k_1)\dd f^+(k_2)], l)$
\State $l \leftarrow l - f^+(i-1)$
\State $r \leftarrow \suc([f_1(i)\dd f_t(i)], l)$
\State $l \leftarrow l - f_{r-1}(i)$
\State \Return \Call{access}{$G, B_r, l \bmod |\gexp(B_r)|$}
\EndIf
\EndFunction
\end{algorithmic}
\end{algorithm}

\begin{example}We show how to access the $\symb$ at position 14 of the string $T = \prod_{i=1}^5\syma^i\symb$. Consider the ISLP $G$ and its auxiliary polynomials computed in Example 1. We start by computing $f^+(2) = 5$. As $l > 5$, we go right in the binary search and compute $f^+(4) = 14$. As $l \leq 14$ we go left, compute $f^+(3) = 9$ and find that $i = 4$. Hence, $T[l]$ lies in the expansion of $A^{i}B = A^4B$ at position $l_1 = l - f^+(i-1) = 5$. Then, we compute $f_1(4) = 4$. As $l_1 > 4$, we turn right and compute $f_2(4) = 5$, finding that $r = 2$. Hence, $T[l]$ lies in the expansion of $B^{i^0} = B^1$ at position $l_2 = l_1 - f_{r-1}(i) = 1$.
\end{example}

\subsection{Extracting substrings}

Once we have accessed $T[l]$, it is possible to output the substring $T[l\dd l+\lambda-1]$ in $O(\lambda+h)$ additional time, as we return from the recursion in Algorithm~\ref{alg:direct-access}. We carry the parameter $\lambda$ of the number of symbols (yet) to output, which is first decremented when we finally arrive at line 3 and find the first symbol, $T[l]$, which we now output immediately. From that point, as we return from the recursion, instead of returning the symbol $T[l]$, we return the number $\lambda$ of symbols yet to output, doing some extra work until $\lambda=0$.
\begin{enumerate}
    \item If we return from line 5, we output $\min(\lambda,|\gexp(C)|)$ symbols from nonterminal $C$, by invoking a new procedure $\textsc{report}(G,C,\lambda)$, which returns the new number $\lambda$ of symbols yet to report; this number is then returned by \textsc{access}.
    \item If we return from line 7, we just return the current value of $\lambda$ to the caller.
    \item If we return from line 13, we must report:
    \begin{enumerate}
        \item $i^{c_r} - \lceil l/|\gexp(B_r)|\rceil$ further copies of $\gexp(B_r)$.
        \item $i^{c_s}$ copies of $\gexp(B_s)$, for $s=r+1,\ldots,t$.
        \item the expansions $\gexp(B_1)^{j^{c_1}}\cdots \gexp(B_t)^{j^{c_t}}$, for $j=i+1,\ldots,k_2$.
    \end{enumerate}
    For each expansion $\gexp(C)$ to report, we invoke $\textsc{report}(G,C,\lambda)$ and update $\lambda$ to the new number of symbols yet to report. We stop if $\lambda=0$. 
\end{enumerate}
Procedure $\textsc{report}(G,C,\lambda)$ outputs $\gexp(C)$ in $O(|\gexp(C)|)$ time if $\lambda \ge |\gexp(C)|$, as it simply traverses the leaves of a tree without unary paths. In this case it returns $\lambda-|\gexp(C)|$. Otherwise, it traverses only the first $\lambda$ leaves of the derivation tree of $C$, in time $O(\lambda+h)$, and returns zero. Once a call to \textsc{report} returns zero, it is never called again; therefore the total time we spend is $O(\lambda+h)$.



\section{Balancing ISLPs} \label{sec:bal}

We show that any $d$-ISLP can be balanced so that its derivation tree is of height $O(\log n)$. 
Actually, we introduce a new type of SLP, which allows us to prove a more general balancing result that subsumes ISLPs.

\begin{definition} A \emph{generalized straight-line program} (GSLP) is an SLP that allows special rules of the form $A \rightarrow x$, where $x$ is a {\em program} (in any Turing-complete language) of length $|x|$ whose output $\mathtt{OUT}(x)$ is a nonempty sequence of variables, none of which can reach $A$. The rule $A \rightarrow x$ contributes $|x|$ to the size of the GSLP; the standard SLP rules contribute as usual. If it holds for all special rules that no variable appears exactly once inside $\mathtt{OUT}(x)$, then the GSLP is said to be \emph{balanceable}.
\end{definition}

We can choose any desired language to describe the programs $x$. Though in principle $|x|$ can be taken as the Kolmogorov complexity of $\mathtt{OUT}(x)$, we will focus on very simple programs and on the asymptotic value of $x$. In particular, RLSLPs allow rules of the form $A \rightarrow B^k$ of size 2, and we can have a program of size $O(1)$ that outputs $k$ copies of $B$; ISLPs allow rules of the form $\prod_{k_1}^{k_2} B_1^{i^c_1}\cdots B_t^{i^{c_t}}$ of size $2+2t$, and we can have a program of size $O(t)$ that writes the corresponding $f^+(k_2)$ symbols. Note that in both cases the GSLP is balanceable as long as special rules satisfy $k>1$ (for RLSLPs), or if $k_1 \neq k_2$ (for ISLPs); otherwise they can be replaced by alternative rules of the same asymptotic size.

We will prove that any balanceable GSLP can be balanced without increasing its asymptotic size. Our proof generalizes that of Ganardi et al.~\cite[Thm.~1.2]{GJL2021} for SLPs in a similar way to how it was extended to balance RLSLPs \cite{NOU2022}. Just as Ganardi et al., in this section we will allow SLPs to have rules of the form $A \rightarrow B_1 \cdots B_t$, of size $t$, where each $B_j$ is a terminal or a nonterminal; this can be converted into a strict SLP of the same asymptotic size.

A \emph{directed acyclic graph} (DAG) is a directed multigraph $D = (V,E)$ without cycles (nor loops). We denote by $|D|$ the number of edges in this DAG. For our purposes, we assume that any DAG has a distinguished node $r$ called the \emph{root}, satisfying that any other node can be reached from $r$ and $r$ has no incoming edges. We also assume that if a node has $k$ outgoing edges, they are numbered from $1$ to $k$, so edges are of the form $(u,i,v)$.  The \emph{sink nodes} of a DAG are the nodes without outgoing edges. The set of sink nodes of $D$ is denoted by $W$. We denote the number of paths from $u$ to $v$ as $\pi(u, v)$, and $\pi(u, V) = \sum_{v \in V}\pi(u, v)$ for a set $V$ of nodes. The number of paths from the root to the sink nodes is $n(D) = \pi(r, W)$.

One can interpret an SLP $G$ generating a string $T$ as a DAG $D$: There is a node for each variable in the SLP, the root node is the initial variable, variables of the form $A \rightarrow a$ are the sink nodes, and a variable with rule $A \rightarrow B_1B_2\dots B_t$ has outgoing edges $(A, i, B_i)$ for $i \in [1,t]$. Note that if $D$ is a DAG representing $G$, then $n(D) = |\gexp(G)| = |T|$.

\begin{definition}{(Ganardi et al. \cite[page 5]{GJL2021})}
Let $D$ be a DAG, and define the pairs $\lambda(v) = (\floor{\log_2\pi(r,v)}, \floor{\log_2\pi(v, W))})$. The \emph{symmetric centroid decomposition (SC-decomposition)} of a DAG $D$ produces a set of edges between nodes with the same $\lambda$ pairs defined as $E_{scd}(D) = \{(u, i, v) \,|\, \lambda(u) = \lambda(v)\}$,
partitioning $D$ into disjoint paths called \emph{SC-paths} (some of them possibly of length 0).
\end{definition}

The set $E_{scd}$ can be computed in $O(|D|)$ time. If $D$ is the DAG of an SLP $G$, then $|D|$ is $O{(|G|)}$. The following lemma justifies the name ``SC-paths''.

\begin{lemma}{(Ganardi et al. \cite[Lemma~2.1]{GJL2021})}\label{lemma:scd}
Let $D = (V , E)$ be a DAG. Then every node has at most one outgoing and at most one incoming edge from $E_{scd}(D)$. Furthermore, every path from the root r to a sink node contains at most $2\log_2 n(D)$ edges that do not belong to $E_{scd}(D)$.
\end{lemma}

Note that the sum of the lengths of all SC-paths is at most the number of nodes of the DAG, or equivalently, the number of variables of the SLP.

The following definition and technical lemma are needed to construct the building blocks of our balanced GSLPs.

\begin{definition}{(Ganardi et al. \cite[page~7]{GJL2021})}
A \emph{weighted string} is a string $T \in \Sigma^*$ equipped with a \emph{weight function} $||\cdot||: \Sigma \rightarrow \mathbb{N} \backslash \{0\}$, which is extended homomorphically. If $A$ is a variable in an SLP $G$, then we write $||A||$ for the weight of the string $\gexp(A)$ derived from $A$.
\end{definition}

\begin{lemma}{(Ganardi et al. \cite[Proposition~2.2]{GJL2021})}\label{lemma:weighted}
For every non-empty weighted string $T$ of length $n$ one can construct in linear time an SLP $G$ generating $T$ with the following properties:
\begin{itemize}
    \item $G$ contains at most $3n$ variables
    \item All right-hand sides of $G$ have length at most 4
    \item $G$ contains suffix variables $S_1 , . . . , S_n$ producing all non-trivial suffixes of $T$
    \item every path from $S_i$ to some terminal symbol $a$ in the derivation tree of $G$ has length at most
$3 + 2(\log_2 ||S_i|| - \log_2 ||a||)$
\end{itemize}
\end{lemma}



\begin{theorem}\label{thm:balancing}Given a balanceable GSLP $G$ generating a string $T$, it is possible to construct an equivalent GSLP $G'$ of size $O(|G|)$ and height $O(\log n)$.\end{theorem}

\begin{proof}Transform the GSLP $G$ into an SLP $H$ by replacing their special rules $A \rightarrow x$ by $A \rightarrow \mathtt{OUT}(x)$, and then obtain the SC-decomposition $E_{scd}(D)$ of the DAG $D$ of $H$. Observe that the SC-paths of $H$ use the same variables of $G$, so it holds that the sum of the lengths of all the SC-paths of $H$ is less than the number of variables of $G$. Also, note that any special variable $A \rightarrow x$ of $G$ is necessarily the endpoint (i.e., the last node of a directed path) of an SC-path in $D$. To see this note that $\lambda(A) \not= \lambda(B)$ for any $B$ that appears in $\mathtt{OUT}(x)$, because $\log_2 \pi(A,W) \geq \log_2 (|\mathtt{OUT}(x)|_B \cdot \pi(B,W)) \ge 1 + \log_2 \pi(B,W)$ where $ |\mathtt{OUT}(x)|_B \ge 2$ because $G$ is balanceable. This implies that the balancing procedure of Ganardi et al. on $H$, which transforms the rules of variables that are not the endpoint of an SC-path in the DAG $D$, will not touch variables that were originally special variables in $G$.

Let $\rho=(A_0 , d_0 , A_1 ), (A_1 , d_1 , A_2 ), \dots , (A_{p-1} , d_{p-1} , A_p)$ be an SC-path of $D$. It holds that for each $A_i$ with $i \in [0\dd p-1]$, in the SLP $H$ its rule goes to two distinct variables, one to the left and one to the right. Thus, for each variable $A_i$, with $i \in [0\dd p-1]$, there is a variable $A_{i+1}'$ that is not part of the path. Let $A_1'A_2'\dots A_p'$ be the sequence of these variables. Let $L = L_1L_2\dots L_s$ be the subsequence of left variables of the previous sequence. Then construct an SLP of size $O(s) \subseteq O(p)$ for the sequence $L$ (seen as a string) as in Lemma \ref{lemma:weighted}, using $|\gexp(L_i)|$ in $H$ as the weight function. In this SLP, any path from the suffix nonterminal $S_i$ to a variable $L_j$ has length at most $3 + 2(\log_2 ||S_i|| - \log_2 ||L_j||)$. Similarly, construct an SLP of size $O(t)\subseteq O(p)$ for the sequence $R = R_1R_2\dots R_t$ of right symbols in reverse order, as in Lemma \ref{lemma:weighted}, but with prefix variables $P_i$ instead of suffix variables. Each variable $A_i$, with $i \in [0\dd p-1]$, derives the same string as $w_{l}A_pw_{r}$, for some suffix $w_{l}$ of $L$ and some prefix $w_{r}$ of $R$. We can find rules deriving these prefixes and suffixes in the SLPs produced in the previous step, so for any variable $A_i$, we construct an equivalent rule of length at most 3. Add these equivalent rules, and the left and right SLP rules to a new GSLP $G'$. Do this for all SC-paths. Finally, add the original terminal variables and special variables (which are left unmodified) of the GSLP $G$, so $G'$ is a GSLP equivalent to $G$. 

The SLP constructed for $L$ has all its rules of length at most 4, and $3s \leq 3p$ variables. The same happens with $R$. The other constructed rules also have a length of at most 3, and there are $p$ of them. Summing over all SC-paths, we have $O(|G|)$ size. The special variables cannot sum up to more than $O(|G|)$ size. Thus, the GSLP $G'$ has size $O(|G|)$.

Any path in the derivation tree of $G'$ is of length $O(\log n)$. To see why, let $A_0,\dots,A_p$ be an SC-path. Consider a path from a variable $A_i$ to an occurrence of a variable that is in the right-hand side of $A_p$ in $G'$. Clearly, this path has length at most 2. Now consider a path from $A_i$ to a variable $A_j'$ in $L$ with $i < j \leq p$. By construction this path is of the form $A_i\rightarrow S_k \rightarrow^* A_j'$ for some suffix variable $S_k$ (if the occurrence of $A_j'$ is a left symbol), and its length is at most $1 + 3 + 2(\log_2 ||S_k|| - \log_2 ||A_j'||) \leq4+2\log_2||A_i||-2\log_2||A_{j}'||$. Analogously, if $A_j'$ is a right variable, the length of the path is bounded by $1 + 3 + 2(\log_2 ||P_k|| - \log_2 ||A_j'||) \leq4+2\log_2||A_i||-2\log_2||A_{j}'||$. Finally, consider a maximal path to a leaf in the derivation tree of $G'$. Factorize it as
$$A_0 \rightarrow^* A_1 \rightarrow^* \dots \rightarrow^* A_{k}$$ 
where each $A_i$ is a variable of $H$ (and also of $G$). Paths $A_i \rightarrow^* A_{i+1}$ are like those defined in the paragraph above, satisfying that their length is bounded by $4+2\log_2||A_i||-2\log_2||A_{i+1}||$. Observe that between each $A_i$ and $A_{i+1}$, in the DAG $D$ there is almost an SC-path, except that the last edge is not in $E_{scd}$. The length of this path is at most
$$\sum_{i=0}^{k-1}(4+2\log_2||A_i||-2\log_2||A_{i+1}||) \leq 4k + 2\log_2||A_0|| - 2\log_2||A_k||$$

By Lemma \ref{lemma:scd}, $k \leq 2\log_2 n$, which yields the upper bound $O(\log n)$. 


To have standard SLP rules of size at most two, delete rules in $G'$ of the form $A \rightarrow B$ (replacing all $A$'s by $B$'s), and note that rules of the form $A \rightarrow BCDE$ or $A \rightarrow BCD$ can be decomposed into rules of length $2$, with only a constant increase in size and depth.  \qed
\end{proof}

By the above theorem, Lemma \ref{lem:dlog}, and because ISLPs can be made balanceable, we obtain the following.

\begin{corollary}
Given a $d$-ISLP $G$ generating a string $T$, it is possible to construct an equivalent $d'$-ISLP $G'$ of size $O(|G|)$, with $d' \le d$, $d' = O(\log n)$,  and height $h' = O(\log n)$.
\end{corollary}

\section{Conclusions}

We have introduced a new extension to straight-line programs (SLPs) and run-length SLPs (RLSLPs) called {\em iterated SLPs (ISLPs)}. ISLPs permit so-called {\em iteration rules} of the form $A \rightarrow \Pi_{i=k_1}^{k_2} B_1^{i^{c_1}}\cdots B_t^{i^{c_t}}$, of size $O(t)$. While it had already been shown that the lower-bound (and unreachable) measure $\delta$, which was text-wise smaller than every preceding measure of repetitiveness, could be outperformed by a {\em reachable} measure (L-systems) on some text families \cite{NU21,NU23}, the size $g_{it}$ of the smallest ISLP generating a text is the first {\em accessible} measure that also outperforms $\delta$ (by the same margin, $O(\delta/\sqrt{n})$ on a text of length $n$). 

With SLPs or RLSLPs representing a text $T[1\dd n]$, an arbitrary symbol of $T$ can be accessed in $O(\log n)$ time. We have shown that, just as SLPs and RLSLPs \cite{GJL2021,NOU2022}, ISLPs can be balanced without asymptotically increasing their space, and used it to devise an algorithm to access any arbitrary text position in time $O(\log^2 n\log\log n)$ within $O(g_{it})$ space. They are also similarly resistant to edits and other text manipulations.


%
%
%
\bibliographystyle{splncs04}
\bibliography{bibliography}

\end{document}